\documentclass[a4paper,11pt]{article}

\usepackage{authblk}

\usepackage{amsmath}
\usepackage{amsthm}
\usepackage{amssymb}

\usepackage{algorithmic}
\usepackage{algorithm}

\usepackage{paralist}

\usepackage{xcolor}

\usepackage{graphicx}

\usepackage[left=1.0in,top=1.0in,right=1.0in,bottom=1.0in,nohead]{geometry}

\usepackage{cite}

\theoremstyle{definition}

\newtheorem{theorem}{Theorem}
\newtheorem{proposition}[theorem]{Proposition}

\newtheorem{lemma}[theorem]{Lemma}

\graphicspath{{figs/}}

\newcommand{\CEIL}[1]{\left\lceil #1\right\rceil}
\newcommand{\BIGP}[1]{\left( #1\right)}
\newcommand{\BIGBP}[1]{\left\{ #1\right\}}

\newcommand{\OP}[1]{\operatorname{#1}}
\newcommand{\MC}[1]{{\mathcal #1}}
\newcommand{\MBB}[1]{{\mathbb #1}}

\newcommand{\INCEDGESET}[1]{{E[#1]}}
\newcommand{\RECEIVED}[2]{D_{#1}(#2)}

\newenvironment{ap_lemma}[1]{\par\noindent{\bf Lemma~#1.} \em}{}

\long\def\longdelete#1{}

\title{$O(f)$ Bi-Approximation for Capacitated Covering with Hard Capacities}

\author[1]{Mong-Jen Kao}
\author[2]{Hai-Lun Tu}
\author[1,2]{D.T. Lee}

\affil[1]{Institute of Information Science, Academia Sinica, Taipei, Taiwan.}

\affil[2]{Department of Computer Science and Information Engineering, National~Taiwan~University,~Taipei,~Taiwan.}

\begin{document}

\maketitle

\begin{abstract}
We consider capacitated vertex cover with hard capacity constraints (VC-HC) on hypergraphs.
%
In this problem we are given a hypergraph $G=(V,E)$ with a maximum edge size $f$.
Each edge is associated with a demand and each vertex is associated with a weight (cost), a capacity, and an available multiplicity.
The objective is to find a minimum-weight vertex multiset such that the demands of the edges can be covered by the capacities of the vertices 
and the multiplicity of each vertex does not exceed its available multiplicity.

\smallskip

In this paper we present an $O(f)$ bi-approximation for VC-HC that gives a trade-off on the number of augmented multiplicity and the cost of the resulting cover.
In particular, we show that, by augmenting the available multiplicity by a factor of $k \ge 2$, a~cover with a cost ratio of $\BIGP{1+\frac{1}{k-1}}(f-1)$ to the optimal cover for the original instance can be obtained.
%
%
This improves over a previous result, which has a cost ratio of $f^2$ via augmenting the available multiplicity by a factor of $f$.
\end{abstract}

\section{Introduction}

The capacitated vertex cover problem with hard capacities (VC-HC) models a demand-to-service assignment scenario generalized from the classical vertex cover problem.
In this problem, we are given a hypergraph $G = (V,E \subseteq 2^V)$ 
with maximum edge size $f$, 
where each $e\in E$ satisfies $|e| \le f$ and is associated with a demand $d_e \in \MBB{R}^{\ge 0}$, and each $v \in V$ is associated with a weight (or cost) $w_v \in \MBB{R}^{\ge 0}$, a capacity $c_v \in \MBB{R}^{\ge 0}$, and an available multiplicity $m_v \in \MBB{Z}^{\ge 0}$.
The objective is to find a vertex multiset, or, cover, represented by a demand assignment function $h\colon E \times V \rightarrow \MBB{R}^{\ge 0}$, such that the following two constraints are met:
\begin{enumerate}
	\item
		$\sum_{v\in e} h_{e,v} \ge d_e$ for all $e\in E$, 
	
	\item
		$x^{(h)}_v \le m(v)$ for all $v \in V$, where $x^{(h)}_v := \CEIL{\sum_{e \colon e\in E, \hspace{2pt} v\in e} h_{e,v} / c_v}$, 
\end{enumerate}
and $\sum_{v \in V} w(v)\cdot x_h(v)$ is minimized.
%
%

\medskip

In this paper, we consider bicriteria approximation for VC-HC with augmented multiplicity constraints.
%
%
In particular, we say that a demand assignment $h$ forms an augmented $(\beta,\gamma)$-cover if it is feasible for the augmented multiplicity function $m'_v := \beta\cdot m_v$ for all $v \in V$ and
the cost ratio is at most $\gamma$ compared to the optimal assignment for the original instance.
In other words, we are allowed to use additional multiplicities of the vertices up to a factor of $\beta$.
%


%
%
%
%


%

%
\paragraph*{Background and Prior Work.}
The capacitated vertex cover generalizes vertex cover in that a demand-to-service assignment model is evolved from the original 0/1 covering model.
This transition was exhibited via several work.

\smallskip

For classical vertex cover, it is known that a $f$-approximation can be obtained by LP rounding and duality~\cite{BarYehuda1981198,doi:10.1137/0211045}.
%
Khot and Regev~\cite{Khot2008335} showed that, assuming the unique game conjecture, approximating this problem to a ratio better than $f - \epsilon$ is NP-hard for any $\epsilon > 0$ and $f \ge 2$.
%
%

\smallskip

Chuzhoy and Naor~\cite{1151271} considered VC-HC on simple graphs with unit edge demands, i.e., $|e| = 2$ and $d_e = 1$ for all $e \in E$. 
They presented a $3$-approximation for the unweighted version of this problem, i.e., $w_v = 1$ for all $v \in V$.
On the contrary, they showed that the weighted version is at least as hard as set cover,
which renders $O(f)$-approximations unlikely to exist even for this simple setting.
%
%
Due to this reason, subsequent work on VC-HC has focused primarily on the unweighted version.

\smallskip

Gandhi~et~al.~\cite{Gandhi:2006:IAA:1740416.1740428} gave a $2$-approximation for unweighted VC-HC with unit edge demand by presenting a refined rounding approach to~\cite{1151271}.
Saha and Khuller~\cite{Saha:2012:SCR:2359332.2359443} considered general edge demands and presented an $O(f)$-approximation for $f$-hypergraphs.
Cheung~et~al.~\cite{doi:10.1137/1.9781611973402.124} presented an improved approach for this problem.
They presented a $\BIGP{1+2/\sqrt{3}}$-approximation for simple graphs and a $2 f$-approximation for $f$-hypergraphs.
The gap of approximation for this problem was recently closed by Kao~\cite{Kao16VCHC}, who presented an $f$-approximation for any $f\ge 2$.
%

\smallskip

Grandoni et al.~\cite{DBLP:journals/siamcomp/GrandoniKPS08} considered weighted VC-HC with unit vertex multiplicity, i.e., $m_v=1$ for all $v \in V$, and augmented multiplicity constraints.
They presented a primal-dual approach that yields an augmented $(2,4)$-cover for simple graphs\footnote{The bicriteria approximation ratio of~\cite{DBLP:journals/siamcomp/GrandoniKPS08} is updated in the context due to the different considered models. In~\cite{DBLP:journals/siamcomp/GrandoniKPS08} each vertex is counted at most once in the cost of the cover, disregarding the number of multiplicities it needs. In our model, however, the cost is weighted over the multiplicities of each vertex.},
which further extends to augmented $(f,f^2)$-cover for $f$-hypergraphs.
This approach does not generalize, however, to arbitrary vertex multiplicities and does not entail further parametric trade-off either.
%

%
%

%


\paragraph*{Further Related Work.}

The capacitated covering problem has been studied in various forms and variations. 
%
%
%
%
When the number of available multiplicities is unlimited, this problem is referred to {soft capacitated vertex cover} (CVC).
%
%
This problem was first considered by Guha~et~al.~\cite{989540}, who gave a $2$-approximation based on primal-dual.
%
%
%
Kao~et~al.~\cite{MJKHLCDTL13,springerlink:10.1007/s00453-009-9336-x,MJKaoDissertation12} studied capacitated dominating set problem and presented a series of results for the complexity and approximability of this problem.
%
%
Bar-Yehuda et al.~\cite{doi:10.1137/080728044} considered partial CVC and presented a $3$-approximation for simple graphs based on local ratio techniques.
%

\smallskip

Wolsey~\cite{Wolsey1982} considered submodular set cover, which includes classical set cover as a special case and which relates to capacitated covering in a simplified form, and presented a $(\ln \max_S f(S) +1)$-approximation.
This approach was generalized by Chuzhoy and Naor~\cite{1151271} to capacitated set cover with hard capacities and unit demands, for which a $(\ln \delta+1)$-approximation was presented, where $\delta$ is the maximum size of the sets.
%
%

%


%

\paragraph*{Our Result and Approach.}

We consider VC-HC with general parameters and present bicriteria approximations that yields a trade-off between the number of augmented multiplicities and the resulting cost.
Our main result is the following bicriteria approximation algorithm:
%

\begin{theorem} \label{thm-augmented-cover}
For any integer $k\ge 2$, we can compute an augmented $\BIGP{k, (1+\frac{1}{k-1})(f-1)}$-cover for weighted VC-HC in polynomial time.
\end{theorem}

This improves over the previous ratio of $(f, f^2)$ in~\cite{DBLP:journals/siamcomp/GrandoniKPS08} and provides a parameter trade-off on the augmented multiplicity and the quality of the solution.
In particular, the cost ratio we obtained for this bi-approximation is bounded within $\frac{3}{2}(f-1)$ for all $k\ge 2$ and converges asymptotically to $f-1$ as $k$ tends to infinity.
%
%

\smallskip

Our algorithm builds on primal-dual charging techniques combined with a flow-based procedure that exploits the duality of the LP relaxation.
The primal-dual scheme we present extends the basic framework from~\cite{springerlink:10.1007/s00453-009-9336-x,989540}, which were designed for the soft capacity model where $m_v = \infty$ for all $v$.
In contrast to the previous result in~\cite{DBLP:journals/siamcomp/GrandoniKPS08}, we employ a different way of handling the dual variables as well as the primal demand assignments that follow.
The seemingly subtle difference entails dissimilar analysis and gives a guarantee that is unavailable via their approach.
%

\smallskip

In particular, for the primal demand assignments,
we use flow-based arguments to deal with pending decisions.
This ensures that the vertices whose multiplicity limits are attained receive sufficient amount of demands to pay for their costs.
The crucial observation in establishing the bicriteria approximation factor is that the feasible regions of the dual LP remains unchanged when the multiplicity constraint is augmented.
Therefore the cost of the solution obtained via the primal-dual approach can be bounded by the optimal cost of the original instance.
Together this gives our bi-approximation result.

\smallskip

%
%
The rest of this paper is organized as follows.
In \S\ref{Preliminary} we formally define VC-HC and introduce the natural LP relaxation and its dual LP for which we will be working with.
For a better flow to present our bicriteria approximation, we first introduce our primal-dual algorithm and the corresponding analysis in~\S\ref{weighted-covering}.
%
In \S\ref{sec-augmented-covering} we establish the bi-approximation approximation ratio and prove Theorem~\ref{thm-augmented-cover}.
%
Finally we conclude in \S\ref{sec-conclusion} with some future directions for related problems.
%
%
Due to the space limit, some of the proofs are omitted from the main content and can be found
in the appendix for further reference.

\section{Problem Statement and LP Relaxation} \label{Preliminary}

Let $G=(V,E)$ denote a hypergraph with vertex set $V$ and edge set $E \subseteq 2^V$ and $f := \max_{e \in E}|e|$ denote the size of the largest hyperedge in $G$.
For any $v \in V$, we use $\INCEDGESET{v}$ to denote the set of edges that are incident to the vertex $v$. Formally, $\INCEDGESET{v} := \{e\colon e\in E \text{ such that } v \in e\}$.
This definition extends to set of vertices, i.e., for any $A \subseteq V$,  i.e., $\INCEDGESET{A} := \bigcup_{v\in A}\INCEDGESET{v}$.
%
%
%
%
%
%

%

\subsection{Capacitated Vertex Cover with Hard Capacities (VC-HC)}
%
%
In this problem we are given a hypergraph $G = (V,E \subseteq 2^V)$, where each $e\in E$ is associated with a demand $d_e \in \MBB{R}^{\ge 0}$ and each $v \in V$ is associated with a weight (or cost) $w_v \in \MBB{R}^{\ge 0}$, a capacity $c_v \in \MBB{R}^{\ge 0}$, and its available multiplicities $m_v \in \MBB{Z}^{\ge 0}$.
%
%

\smallskip

By a demand assignment we mean a function $h\colon E \times V \rightarrow {\mathbb{Z}}^{\ge 0}$, where $h_{e,v}$ denotes the amount of demand that is assigned from edge $e$ to vertex $v$.
For any $v \in V$, we use $\RECEIVED{h}{v}$ to denote the total amount of demand vertex $v$ has received in $h$, i.e., $\RECEIVED{h}{v} = \sum_{e \in \INCEDGESET{v}}h_{e,v}$.

\smallskip

The corresponding multiplicity function, denoted $x^{(h)}$, is defined to be 
$x^{(h)}_v = \CEIL{\RECEIVED{h}{v} / c_v}$.
%
%
%
%
A demand assignment $h$ is \emph{feasible} if 
		$\sum_{v \in e}h_{e,v} \ge d_e$ for all $e\in E$ and
%
		$x^{(h)}_v \le m_v$ for all $v \in V$.
%
In other words, the demand of each edge is fully-assigned to (fully-served by) its incident vertices and the multiplicity of each vertex does not exceed its available multiplicities.
%
%
%
The \emph{weight (cost)} of 
$h$, denoted $w(h)$, is defined to be $\sum_{v \in V} w_v\cdot x^{(h)}_v$.

\smallskip

%

%

%
Given an instance $\Pi = (V, E, d_e, w_v, c_v, m_v)$ 
as described above, the problem of VC-HC is to compute a feasible demand assignment $h$ such that $w(h)$ is minimized.
Without loss of generality, we assume that the input graph $G$ admits a feasible demand assignment.\footnote{By selecting all of the available multiplicities, the feasibility of $G$ can be checked via a max-flow computation.}
%


%


\paragraph*{Augmented Cover.}
Let $\Pi = (V, E, d_e, w_v, c_v, m_v)$ be an instance for VC-HC.
For any integral $\beta \ge 1$, we say that a demand assignment $h$ forms an augmented $(\beta, \gamma)$-cover if
\begin{enumerate}[\qquad (1)]
	\item
		$\sum_{v \in e}h_{e,v} \ge d_e$ for all $e \in E$.
		
	\item
		$x^{(h)}_v \le \beta\cdot m_v$ for all $v \in V$.
		
	\item
		$w(h) \le \gamma \cdot \min_{h' \in \MC{F}}w(h')$, where $\MC{F}$ is the set of feasible demand assignments for $\Pi$.
\end{enumerate}
%
%
%

%


\subsection{LP Relaxation and the Dual LP}
Let $\Pi = (V, E, d_e, w_v, c_v, m_v)$ be the input instance of VC-HC.
The natural LP relaxation of VC-HC for the instance $\Pi$ is given below in LP(\ref{ILP_cdh}). 
%
%
The first three inequalities model the feasibility constraints of a demand assignment and its corresponding multiplicity function.
The fourth inequality states that the multiplicity of a vertex cannot be zero if any demand gets assigned to it.
This seemingly unnecessary constraint is required in giving a bounded integrality gap for this LP relaxation.

\smallskip

%
\begin{center}
\fbox
{
\hspace{12pt}
\begin{minipage}{0.75\textwidth}
\vspace{-4pt}
\begin{align}
\mathrm{Minimize} \quad \sum_{v \in V}w_v\cdot x_v & & & \phantom{\forall e \hspace{2.2cm} \in E} \label{ILP_cdh}
\end{align}
\vspace{-14pt}
\begin{align*}
& \sum_{v \in e}h_{e,v} \hspace{2pt} \ge \hspace{2pt} d_e, & & \forall e \in E \notag \\
& c_v\cdot x_v - \sum_{e \in E[v]} h_{e,v} \hspace{2pt} \ge \hspace{2pt} 0, & & \forall v \in V \notag \\
& x_v \hspace{2pt} \le \hspace{2pt} m_v, & & \forall v \in V \notag \\[2pt]
& d_e \cdot x_v \hspace{2pt} - \hspace{2pt} h_{e,v} \hspace{2pt} \ge \hspace{2pt} 0,  & & \forall e \in E, \enskip v \in e  \notag \\[2pt]
& x_v, \hspace{2pt} h_{e,v} \hspace{2pt} \ge \hspace{2pt} 0, & & \forall e \in E, \enskip v \in e \notag
\end{align*}
\vspace{-14pt}
\end{minipage}\enskip}
\end{center}
%

\smallskip

The dual LP for the instance $\Pi$ is given below in~LP(\ref{ILP_dual_cdh}).
A solution $\Psi = (y_e, z_v, g_{e,v}, \eta_v)$ to this LP can be interpreted as an extended packing LP as follows:
We want to raise the values of $y_e$ for all $e \in E$. 
However, the value of each $y_e$ is constrained by $z_v$ and $g_{e,v}$ that are further constrained by $w_v$ for each $v \in e$.
The variable $\eta_v$ provides an additional degree of freedom in this packing program in that it allows higher values to be packed into $y_e$ in the cost of a reduction in the objective value.
Note that, this exchange does not always yield a better lower-bound for the optimal solution.
In this paper we present an extended primal-dual scheme to handle this flexibility.


\begin{center}
\fbox
{
\hspace{12pt}
\begin{minipage}{0.85\textwidth}
\vspace{-2pt}
\begin{align}
\mathrm{Maximize} \quad \sum_{e \in E} d_e\cdot y_e \hspace{2pt} - \hspace{2pt} \sum_{v \in V} m_v \cdot \eta_v & & & \phantom{\forall e \hspace{2.2cm} \in E} \label{ILP_dual_cdh}
\end{align}
\vspace{-14pt}
\begin{align*}
& c_v\cdot z_v \hspace{2pt} + \hspace{2pt} \sum_{e \in \INCEDGESET{v}}d_e \cdot g_{e,v} \hspace{2pt} - \hspace{2pt} \eta_v \hspace{2pt} \le \hspace{2pt} w_v, & & \hspace{-10pt} \forall v \in V \notag \\
& y_e \hspace{2pt} \le \hspace{2pt} z_v \hspace{2pt} + \hspace{2pt} g_{e,v}, & & \hspace{-10pt} \forall v \in V, \enskip e \in \INCEDGESET{v} \notag \\[3pt]
& y_e, \hspace{2pt} z_v \hspace{2pt} g_{e,v} \hspace{2pt} \eta_v \hspace{2pt} \ge \hspace{2pt} 0, & & \hspace{-10pt} \forall v \in V, \enskip e \in \INCEDGESET{v} \notag
\end{align*}
\vspace{-12pt}
\end{minipage}\enskip}
\end{center}
%

For the rest of this paper, we will use $\OP{OPT}(\Pi)$ to denote the cost of optimal solution for the instance $\Pi$.
Since the optimal value of the above LPs gives a lower-bound on $\OP{OPT}(\Pi)$ which we will be working with, we also use $\OP{OPT}(\Pi)$ to denote their optimal value in the context.

\section{A Primal-Dual Schema for VC-HC} \label{weighted-covering}

In this section we present our extended primal-dual algorithm for VC-HC.
The algorithm we present extends the framework developed for the soft capacity model~\cite{springerlink:10.1007/s00453-009-9336-x,989540}.
In the prior framework, the demand is assigned immediately when a vertex from its vicinity gets saturated.
In our algorithm, we keep some of decisions pending until we have sufficient capacity for the demands.
In contrast to the primal-dual scheme used in~\cite{DBLP:journals/siamcomp/GrandoniKPS08}, which always stores dual values in $g_{e,v}$, we store the dual values in both $g_{e,v}$ and $z_v$, depending on the amount of unassigned demand $v$ possesses in its vicinity.
This ensures that, the cost of each multiplicity is charged only to the demands it serves.

\smallskip

To obtain a solid bound for this approach, however, we need to guarantee that the vertices whose multiplicity limits are attained receive sufficient amount of demands to charge to.
This motivates our flow-based procedure $\OP{Self-Containment}$ for dealing with the pending decisions.
During this procedure, a natural demand assignment is also formed.
%
%
%
%

%
%
%
%


\subsection{The Algorithm}
\label{subsec-primal-dual}

In this section we present our extended primal-dual algorithm {\sc Dual-VCHC}.
This algorithm takes as input an instance $\Pi = (V,E,d,w,c,m)$ of VC-HC and outputs a feasible primal demand assignment $h$ together with a feasible dual solution $\Psi = \BIGP{y_v, z_v, g_{e,v}, \eta_v}$ for $\Pi$.

\smallskip

The algorithm starts with an initial zero dual solution and eventually reaches a locally optimal solution. During the process, the values of the dual variables in $\Psi$ are raised gradually and some inequalities will meet with equality.
We say that a vertex $v$ is \emph{saturated} if the inequality $c_v \cdot z_v + \sum_{e \in \INCEDGESET{v}}d_e \cdot g_{e,v} - \eta_v \le w_v$ is met with equality.

\smallskip

Let $E^\phi := \{e: e\in E, d_e > 0\}$ be the set of edges with non-zero demand and $V^\phi := \{v: v\in V, m_v\cdot c_v > 0\}$ be the set of vertices with non-zero capacity.
For each $v \in V$, we use $d^\phi(v) = \sum_{e \in \INCEDGESET{v} \cap E^\phi }d_e$ to denote the total amount of demand in $\INCEDGESET{v} \cap E^\phi$.
For intuition, $E^\phi$ contains the set of edges whose demands are not yet processed nor assigned, and $V^\phi$ corresponds to the set of vertices that have not yet saturated.

\smallskip

In addition, we maintain a set $S$, initialized to be empty, to denote the set of vertices that have saturated and that have at least one incident edge in $E^\phi$.
Intuitively, $S$ corresponds to vertices with pending assignments.
%

\medskip

The algorithm works as follows.
Initially all dual variables in $\Psi$ and the demand assignment $h$ are set to be zero.
We raise the value of the dual variable $y_e$ for each $e \in E^\phi$ simultaneously at the same rate.
To maintain the dual feasibility, as we increase $y_e$, either $z_v$ or $g_{e,v}$ has to be raised for each $v \in e$.
If $d^\phi(v) \le c_v$, then we raise $g_{e,v}$. 
Otherwise, we raise $z_v$.
In addition, for all $v \in e \cap S$, we \emph{raise} $\eta_v$ to the extent that keeps $v$ saturated.

\smallskip

When a vertex $u \in V^\phi$ becomes saturated, it is removed from $V^\phi$.
Then we invoke a recursive procedure $\OP{Self-Containment}(S\cup\{u\},u)$, which we describe in the next paragraph, to compute a pair $(S',h')$, where 
\begin{itemize}
	\item
		$S'$ is a maximal subset of $S\cup\BIGBP{u}$ whose capacity, if chosen, can fully-serve the demands in $\INCEDGESET{S'}\cap E^\phi$, and 
		
	\item
		$h'$ is the corresponding demand assignment function (from $\INCEDGESET{S'}\cap E^\phi$ to $S'$). 
\end{itemize}
If $S' = \emptyset$, then we leave the assignment decision pending and add $u$ to $S$. 
Otherwise, $S'$ is removed from $S$ and $\INCEDGESET{S'}$ is removed from $E^\phi$.
In addition, we add the assignment $h'$ to final assignment $h$ to be output.
This process repeats until $E^\phi = \emptyset$.
Then the algorithm outputs $h$ and $\Psi$ and terminates.
A pseudo-code for this algorithm can be found in Figure~\ref{algorithm_greedy_charging}.

\smallskip

We also note that, the particular vertex to saturate in each iteration is the one with the smallest value of $w^\phi(v) / \min\{c_v, d^\phi(v)\}$, where $w^\phi(v) := w_v- \BIGP{c_v \cdot z_v + \sum_{e\in \INCEDGESET{v}}d_e \cdot g_{e,v} - \eta_v}$ denotes the current slack of the inequality associated with $v \in V^\phi$.

\smallskip


%
\paragraph*{The Procedure $\OP{Self-Containment}(A,u)$.}
%
%
%
In the following we describe the recursive procedure $\OP{Self-Containment}(A,u)$.
%
It takes as input a vertex subset $A \subseteq V$ and a vertex $u \in V$ and outputs a pair $(S',\tilde{h}')$, where $S'$ is a maximal subset of $A$ whose capacity is sufficient to serve the unassigned demands in its vicinity, and $h'$ is the corresponding demand assignment.

\smallskip

First we define a directed flow-graph $\MC{G}(A)$ with a source $s^+$ and a sink $s^-$ for the vertex set $A$ as follows.
Excluding the source $s^+$ and the sink $s^-$, $\MC{G}(A)$ is a bipartite graph induced by $\INCEDGESET{A} \cap E^\phi$ and $A$.
%
For each $e \in \INCEDGESET{A} \cap E^\phi$, we have a vertex $\tilde{e}$ and an edge $(s^+,\tilde{e})$ in $\MC{G}$. 
Similarly, for each $v \in A$ we have a vertex $\tilde{v}$ and an edge $(\tilde{v},s^-)$.
For each $v \in A$ and each $e \in \INCEDGESET{v} \cap E^\phi$, we have an edge $(\tilde{e}, \tilde{v})$ in $\MC{G}$.

\smallskip

The capacity of each edge is defined as follows.
For each $e \in \INCEDGESET{A} \cap E^\phi$, the capacity of $(s^+,\tilde{e})$ is set to be $d_e$.
For each $v \in A$, the capacity of $(\tilde{v},s^-)$ is set to be $m_v\cdot c_v$.
The capacities of the remaining edges are unlimited.

\smallskip

The procedure $\OP{Self-Containment}$ works as follows.
If $u \in A$, then it computes the max-flow $\tilde{h}$ for $\MC{G}(A)$ with the additional constraint that 
$\tilde{h}(\tilde{u},s^-)$ is minimized among all max-flows for $\MC{G}(A)$.\footnote{This criterion can be achieved by imposing an additional constraint when computing the augmenting paths.}
If $u \notin A$, then it simply computes any max-flow $\tilde{h}$ for $\MC{G}(A)$.
Let $$S' = \BIGBP{v \colon v \in A \text{ such that } \tilde{h}(s^+,\tilde{e}) = d_e \text{ for all $e \in \INCEDGESET{v}\cap E^\phi$}}$$ be the subset of $A$ that is able to serve the demand in $\INCEDGESET{S'} \cap E^\phi$.
%
If $S' = A$ or $S' = \emptyset$, then it returns $(S', \tilde{h'})$, where $\tilde{h'}$ is the demand assignment induced by $\tilde{h}$.
Otherwise it returns $\OP{Self-Containment}(S',u)$.
\subsection{Properties of {\sc Dual-VCHC}}
\label{sec-properties-primal-dual}
%
Below we derive basic properties of our algorithm.
Since the algorithm keeps the constraints feasible when increasing the dual variables, we know that $\Psi$ is feasible for the dual LP for $\Pi$.
In the following, we first show that $h$ is a feasible demand assignment for $\Pi$ as well.
Then we derive properties we will be using when establishing the bi-approximation factor next section.

\paragraph*{Feasibility of the demand assignment $h$.}
We begin with 
procedure $\OP{Self-Containment}$.
%
%
%
%
Let $(S', \tilde{h}')$ be the pair returned by $\OP{Self-Containment}(S\cup\{u\},u)$.
The following lemma shows that $S'$ is indeed maximal.
%

\begin{lemma} \label{lemma_local_fully_serve}
If there exists a $B \subseteq S \cup \{u\}$ such that $B$ can fully-serve the demand in $\INCEDGESET{B} \cap E^\phi$, then $B \subseteq S'$.
\end{lemma}

\begin{proof}
Let $S_1, S_2, \ldots, S_k$, where $S_1 = S \cup \{u\} \supset S_2 \supset \ldots \supset S_k = S'$, denote the input of the procedure $\OP{Self-Containment}(S\cup\{u\},u)$ in each recursion.

\smallskip

Below we argue that $B \subseteq S_i$ implies that $B \subseteq S_{i+1}$ for all $1\le i<k$.
Let $\tilde{h}_B$ denote a maximum flow for the flow graph $\MC{G}(B)$.
Since $B$ can fully-serve the demand in $\INCEDGESET{B}\cap E^\phi$, we know that $\tilde{h}_B(s^+, \tilde{e}) = d_e$ for all $e \in \INCEDGESET{B} \cap E^\phi$.

\smallskip

Consider the flow function computed by $\OP{Maxflow}(\MC{G}(S_i),u)$ and denote it by $\tilde{h}_i$.
If $\tilde{h}_i(s^+, \tilde{e}) < d_e$ for some $e \in \INCEDGESET{B} \cap E^\phi$, then we embed $\tilde{h}_B$ into $\tilde{h}_i$, i.e., cancel the flow from $\INCEDGESET{B} \cap E^\phi$ to $B$ in $\tilde{h}_i$ and replace it by $\tilde{h}_B$.
%
We see that the resulting flow strictly increases and remains valid for $\MC{G}(S_i)$, which is a contradiction to the fact that $\tilde{h}_i$ is a maximum flow for $\MC{G}(S_i)$.
Therefore, we know that $\tilde{h}_i(s^+, \tilde{e}) = d_e$ for all $e \in \INCEDGESET{B} \cap E^\phi$ and the vertices of $B$ must be included in $S_{i+1}$.
This show that $B \subseteq S_i$ for all $1\le i\le k$.
\end{proof}

\noindent
The following lemma states the feasibility of this primal-dual process.

\begin{lemma} \label{lemma-primal-dual-feasibility}
$E^\phi$ becomes empty in polynomial time.
Furthermore, the assignments computed by $\OP{Self-Containment}$ during the process form a feasible demand assignment.
\end{lemma}

\paragraph*{The cost incurred by $h$.}
Below we consider the cost incurred by the partial assignments computed by $\OP{Self-Containment}$.
Let $V_S$ denote the set of vertices that have been included in the set $S$.
For any vertex $v$ that has saturated, we use $(S'_v, h'_v)$ to denote the particular pair returned by $\OP{Self-Containment}$ such that $v \in S'_v$.
Note that, this pair $(S'_v, h'_v)$ is uniquely defined for each $v$ that has saturated.
%
%
%
%
Therefore, we know that $h_{e,v} = (h'_v)_{e,v}$ holds for any $e \in E[v]$.

\smallskip

In the rest of this section, we will simply use $h_{e,v}$ when it refers to $(h'_v)_{e,v}$ for simplicity of notations.
Recall that $\RECEIVED{h'_v}{v}$ denotes the amount of demand $v$ receives in $h'_v$.
We have the following proposition for the dual solution $\Psi = (y_e, z_v, g_{e,v}, \eta_v)$, which follows directly from the way the dual variables are raised.

\begin{proposition} \label{lemma_u_notin_s_heavy_unicost}
For any $v \in V$ such that $d^\phi(v) > c_v$ when saturated, the following holds:
\begin{itemize}
	\item
		$z_v = y_e$ for all $e \in E[v]$ with $h_{e,v} > 0$.
		
	\item
		$\eta_v > 0$ only when $v \in V_S$.
\end{itemize}
\end{proposition}

\smallskip

The following lemma gives the properties for vertices in $V_S$.
%

\begin{figure*}[tp]
\centering
\fbox
{\includegraphics[scale=0.9]{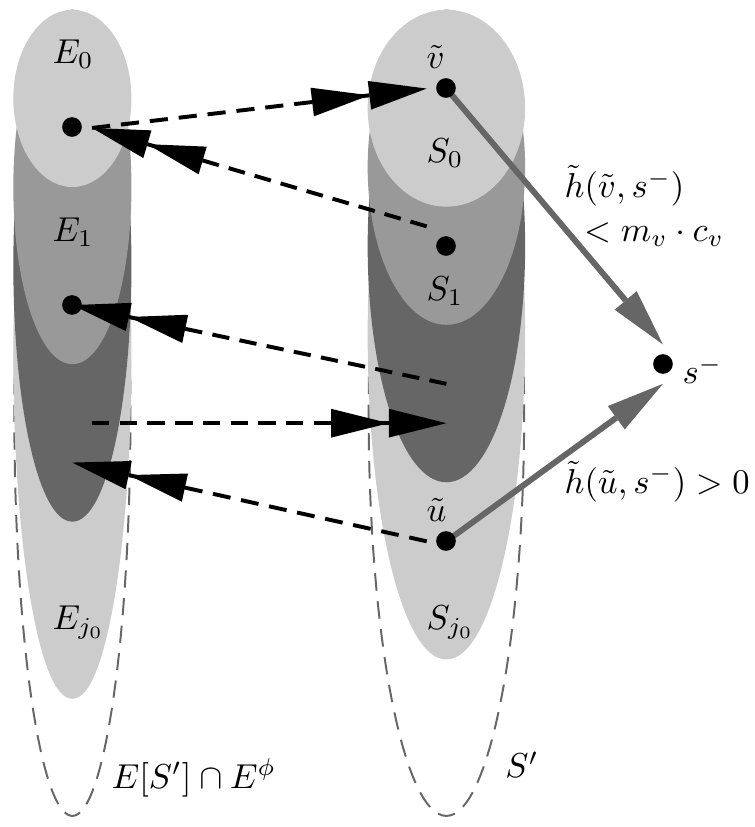}}
\caption{Alternating paths in the flow-graph $\MC{G}(S')$.}
\label{fig-S-augmenting-path}
\end{figure*}

\begin{lemma} \label{lemma_u_in_s_unicost}
For any $v \in V_S$, we have 
\begin{enumerate}
	\item
		$\RECEIVED{h'_v}{v} = m_v\cdot c_v$. 
		
	\item
		$w_v\cdot m_v = \RECEIVED{h'_v}{v}\cdot y_e - m_v\cdot \eta_v$ for all $e \in \INCEDGESET{v}$ such that $h_{e,v} > 0$.
\end{enumerate}
\end{lemma}

\begin{proof}
%
First we prove that $\RECEIVED{h'_v}{v} < m_v\cdot c_v$.
Without loss of generality, we assume that $m_v \ge 1$ and $\RECEIVED{h'_v}{v} < m_v\cdot c_v$ for a contradiction.

\smallskip

Consider the iteration for which the vertex $v$ was removed from $S$
and let $u$ be the vertex that becomes saturated in that iteration.
%
%
By Lemma~\ref{lemma_local_fully_serve}, we know that in the beginning of that iteration, $\nexists B \subseteq S$ such that $B$ can fully-serve $\INCEDGESET{B} \cap E^\phi$.
Therefore it follows that $u \in S'_v$, for otherwise $S'_v$ would have been removed from $S$ in the previous iteration.

\smallskip

Consider the flow-graph $\MC{G}(S'_v)$ and the max-flow $\tilde{h}'_v$ to which $h'_v$ corresponds.
We know that $\tilde{h}'_v(\tilde{e},\tilde{u}) = 0$ for all $e \in \INCEDGESET{v} \cap E^\phi$, for otherwise we have an alternating path $\tilde{u} \rightarrow \tilde{e} \rightarrow \tilde{v}$ so that we can reroute the flow $\tilde{e} \rightarrow \tilde{u} \rightarrow s^-$ to $e \rightarrow \tilde{v} \rightarrow s^-$, which is a contradiction to the fact that the max-flow we compute is the one that minimizes the flow from $\tilde{u}$ to $s^-$.

\smallskip

Let $S_0 := \{v\}$ and $E_0 := \INCEDGESET{v} \cap E^\phi$.
For $i \ge 1$, consider the sets $S_i$ and $E_i$ defined as
$$S_i := \bigcup_{e \in E_{i-1}}\BIGBP{v'\colon v' \in e \cap S'_v} \enskip \text{and} \enskip E_i := \INCEDGESET{S_i} \cap E^\phi.$$
%
%
%
%
Note that, $u \notin S_i$ implies that $S_i \subsetneq S_{i+1}$, for otherwise $S_i$ would be a subset of $S$ that can fully-serve $\INCEDGESET{S_i} \cap E^\phi$ since the beginning of the iteration, a contradiction to Lemma~\ref{lemma_local_fully_serve}.
Therefore $u \in S_j$ for some $j \ge 1$ since $|S_i| \le |S'_v| < \infty$. Let $j_0$ be the smallest integer such that $u \in S_{j_0}$.
By definition we have $S_0 \subsetneq S_1 \subsetneq \ldots \subsetneq S_{j_0} \subseteq S'_v$.
This corresponds to an alternating path to which we can reroute the flow from $u$ to $v$, a contradiction.
See also Fig.~\ref{fig-S-augmenting-path} for an illustration.
Therefore we have $\RECEIVED{h'_v}{v} = m_v\cdot c_v$.

\smallskip

For the second half of this lemma, since $v \in V_S$, we know that $d^\phi(v) > c_v$ before it gets saturated.
Therefore, by Proposition~\ref{lemma_u_notin_s_heavy_unicost},
we know that $y_e = z_v$ holds for all $e \in \INCEDGESET{v}$ such that $h_{e,v} > 0$.
It follows that $w_v = c_v\cdot z_v - \eta_v = c_v\cdot y_e - \eta_v$
%
and $w_v\cdot m_v = \RECEIVED{h'_v}{v}\cdot y_e - m_v \cdot \eta_v$ as claimed.
%
\end{proof}

%

%



%
The following auxiliary lemma, which is carried over from the previous primal-dual framework, shows that, for any vertex $v$ with $d^\phi(v) \le c_v$ when saturated, we can locate at most $c_v$ units of demands from $\INCEDGESET{v}$ such that their dual value pays for $w_v$.
%
This statement holds intuitively since $v$ is saturated.
%
		
\begin{lemma} \label{claim-light-demands}
For any $v \in V$ with $d^\phi(v) \le c_v$ when saturated, 
we can compute a function $\ell_v\colon \INCEDGESET{v} \rightarrow {\mathbb R}^{\ge 0}$ such that the following holds:
\begin{enumerate}[\quad (a)]
	\item
		$0\le h_{e,v} \le \ell_v(e) \le d_e$, for all $e \in \INCEDGESET{v}$.
	\item
		$\sum_{e \in \INCEDGESET{v}}\ell_v(e) \le c_v$.
	\item
		$\sum_{e \in \INCEDGESET{v}}\ell_v(e)\cdot y_e = w_v$.
\end{enumerate}
\end{lemma}

Intuitively, Proposition~\ref{lemma_u_notin_s_heavy_unicost} and Lemma~\ref{lemma_u_in_s_unicost} provide a solid upper-bound for vertices whose capacity is fairly used. 
However, we remark that, this approach does not yield a solid guarantee for vertices whose capacity is barely used, i.e., $\RECEIVED{h'_v}{v} \ll c_v$.
The reason is that the demand that is served (charged) by vertices that have been included in $S$, i.e., those discussed in Lemma~\ref{lemma_u_in_s_unicost}, cannot be charged again since their dual values are inflated during the primal-dual process.
%
%


%

%



\section{Augmented $\BIGP{k, (1+\frac{1}{k-1})(f-1)}$-Cover} \label{sec-augmented-covering}

In this section we establish the following theorem: 

\begin{theorem} \label{thm-augmented-cover-restate}
For any integer $k\ge 2$, we can compute an augmented $\BIGP{k, (1+\frac{1}{k-1})(f-1)}$-cover for VC-HC in polynomial time.
\end{theorem}

Let $\Pi = (V,E,d,w,c,m)$ be the input instance.
Let $m'_v := k\cdot m_v$ denote the augmented multiplicity function for each $v \in V$.
We invoke algorithm {\sc Dual-VCHC} on the instance $\Pi' = (V,E,d,w,c,m')$.
Let $h$ be the demand assignment and $\Psi = (y,z,g,\eta)$ be the dual solution output by the algorithm for $\Pi'$.

\smallskip

The following observation is crucial in establishing the bi-approximation ratio: The dual solution $\Psi$, which was computed for instance $\Pi'$, is also feasible for input instance $\Pi$.

\begin{lemma} \label{lemma-dual-identical-feasible-region}
$\Psi$ is feasible for LP(\ref{ILP_dual_cdh}) with respect to $\Pi$.
In other words, we have $$\sum_{e \in E}d_e \cdot y_e - \sum_{v \in V}m_v \cdot \eta_v \hspace{3pt} \le \hspace{3pt} OPT(\Pi).$$
\end{lemma}

\begin{proof}
The statements follow directly since LP(\ref{ILP_dual_cdh}) has the same feasible region for $\Pi$ and $\Pi'$.
\end{proof}

It is also worth mentioning that, the assignment $h$ computed by {\sc Dual-VCHC} already gives an augmented $\BIGP{k, (1+\frac{1}{k-1})f}$-cover.
To obtain our claimed ratio, however, we further modify some of the demand assignments in $h$ to achieve better utilization on the residue capacity of the vertices.
Below we describe this procedure and establish the bi-approximation ratio.
%

\smallskip

Let $V_S$ denote the set of vertices that have been included in $S$.
For each $v \in V$ such that $\RECEIVED{h}{v} < c_v$, let $\ell_v$ denote the function given by Lemma~\ref{claim-light-demands} with respect to $v$.
%
%
%
%
%
We use $h^*$ to denote the resulting assignment to obtain, where $h^*$ is initialized to be $h$.
For each $e \in E$, we repeat the following operation until no such vertex pair can be found:
\begin{itemize}
	\item
		Find a vertex pair $u \in e \setminus V_S$ and $v \in e$ such that 
		$$\begin{cases}h^*_{e,u} > 0, \\[2pt] \RECEIVED{h^*}{u} > c_u,\end{cases} \text{and} \quad \begin{cases}\RECEIVED{h}{v} < c_v, \\[2pt] h^*_{e,v} < \ell_v(e). \end{cases}$$
		
		Then
		reassign $\min\big\{ \hspace{3pt} h^*_{e,u}, \enskip \ell_v(e) - h^*_{e,v} \hspace{3pt} \big\} \enskip \text{units of demand of $e$ from $u$ to $v$.}$
		In particular, we set \enskip $\begin{cases} h^*_{e,u} = h^*_{e,u} - R_{u,v}, \\ h^*_{e,v} = h^*_{e,v} + R_{u,v}, \end{cases}$ where $R_{u,v} := \min\big\{ h^*_{e,u}, \enskip \ell_v(e) - h^*_{e,v} \big\}$.
\end{itemize}

Intuitively, in assignment $h^*$ if some demand is currently assigned to a vertex in $V \setminus V_S$ that requires multiple multiplicities, then we try to reassign it to vertices that have surplus residue capacity (according to the function $\ell_v$) to balance the load.
Note that, in this process we do not use additional multiplicities of the vertices, and the reassignments are performed only between vertices not belonging to $V_S$.

\smallskip

The following lemma shows that, the cost incurred by vertices in $V \setminus V_S$ can be distributed to the dual variables of the edges.

\begin{lemma} \label{lemma-VC-RHC-heavy-light-mixed-bound}
We have $$\sum_{v \in V \setminus V_S}w_v\cdot x^{(h^*)}_v \le \BIGP{f-1}\cdot \sum_{v \in V_S}\sum_{e \in \INCEDGESET{v}}h^*_{e,v}\cdot y_e + f\cdot\sum_{v \in V \setminus V_S}\sum_{e \in \INCEDGESET{v}}h^*_{e,v}\cdot y_e.$$
\end{lemma}

The following lemma provides a lower bound for $\OP{OPT}(\Pi)$ in terms of the net sum of the dual values over the edges.
%
%

\begin{lemma} \label{lemma-VC-RHC-opt-lower-bound}
We have
$$\sum_{e \in E} d_e\cdot y_e \enskip \le \enskip \frac{k}{k-1} \cdot \OP{OPT}(\Pi).$$
\end{lemma}

\begin{proof}
For each $v \in V_S$, by Lemma~\ref{lemma_u_in_s_unicost} we have $\sum_{e \in \INCEDGESET{v}}h^*_{e,v} = m'_v\cdot c_v = k\cdot m_v\cdot c_v.$
Furthermore, by the way how $\eta_v$ is raised, we know that $\eta_v \hspace{2pt} \le \hspace{2pt} c_v\cdot z_v \hspace{2pt} = \hspace{2pt} c_v\cdot y_e$ holds for all $e \in \INCEDGESET{v}$ such that $h^*_{e,v} > 0$.
Therefore, it follows that 
\begin{equation}
m_v\cdot \eta_v \enskip \le \enskip m_v\cdot c_v\cdot y_e \enskip \le \enskip \frac{1}{k}\cdot \sum_{e \in \INCEDGESET{v}}h^*_{e,v} \cdot y_e.
\label{ieq-opt-lower-bound}
\end{equation}
By Inequality~(\ref{ieq-opt-lower-bound}) and Lemma~\ref{lemma-dual-identical-feasible-region}, it follows that
\begin{equation}
\sum_{e \in E}\BIGP{d_e - \frac{1}{k} \cdot \sum_{v\in e \cap V_S}h^*_{e,v}}\cdot y_e \enskip \le \enskip \OP{OPT}(\Pi).
\label{ieq-opt-lower-bound-2}
\end{equation}
Therefore,
\begin{align*}
\sum_{e \in E} d_e \cdot y_e \enskip = \enskip \sum_{v \in V} \sum_{e \in \INCEDGESET{v}}h^*_{e,v}\cdot y_e \enskip 
& \le \enskip \sum_{e \in E} \BIGP{\frac{k}{k-1} \cdot d_e -  \frac{1}{k-1}\cdot \sum_{v \in e \cap V_S}h^*_{e,v}}\cdot y_e \\
& = \enskip \frac{k}{k-1}\cdot\BIGP{\sum_{e \in E}\BIGP{d_e - \frac{1}{k}\cdot\sum_{v\in V_S\cap e}h^*_{e,v}}\cdot y_e} \\
& \le \enskip \frac{k}{k-1} \cdot \OP{OPT}(\Pi),
\end{align*}
where the last inequality follows from Inequality~(\ref{ieq-opt-lower-bound-2}).
\end{proof}

\noindent
In the following we establish the bi-criteria approximation factor and prove Theorem~\ref{thm-augmented-cover-restate}.

\begin{lemma} \label{lemma-VC-RHC-overall-bound}
We have $$w(h^*) \le \BIGP{1+\frac{1}{k-1}}\cdot\BIGP{f-1}\cdot OPT(\Pi)$$ for any integer $k \ge 2$.
\end{lemma}

\begin{proof}
By Lemma~\ref{lemma_u_in_s_unicost}, we have $\RECEIVED{h}{v} = m'_v\cdot c_v = k\cdot m_v\cdot c_v$ for any $v \in V_S$.
Therefore, $$w_v\cdot x^{(h^*)}_v = \BIGP{c_v \cdot z_v - \eta_v}\cdot k\cdot m_v =  \sum_{e \in \INCEDGESET{v}}h^*_{e,v}\cdot y_e - k\cdot m_v\cdot \eta_v.$$
%
%
Applying Lemma~\ref{lemma-VC-RHC-heavy-light-mixed-bound}, we obtain
\begin{align*}
w(h^*) \enskip & = \enskip \sum_{v \in V_S}w_v \cdot x^{(h^*)}_v + \sum_{v \in (V \setminus V_S)}w_v \cdot x^{(h^*)}_v \\
& \le \enskip \BIGP{\sum_{v \in V_S}\sum_{e \in \INCEDGESET{v}}h^*_{e,v}\cdot y_e - k \cdot \sum_{v \in V} m_v\cdot \eta_v} \\
& \qquad \qquad + \BIGP{(f-1)\cdot\sum_{v \in V_S}\sum_{e \in \INCEDGESET{v}}h^*_{e,v}\cdot y_e + f\cdot \sum_{v \in (V \setminus V_S)}\sum_{e \in \INCEDGESET{v}}h^*_{e,v}\cdot y_e} \\[6pt]
& = \enskip f \cdot \sum_{v \in V} \sum_{e \in \INCEDGESET{v}}h^*_{e,v}\cdot y_e - k \cdot \sum_{v \in V} m_v\cdot \eta_v \\[6pt]
& = \enskip k \cdot \BIGP{\sum_{v \in V}\sum_{e \in \INCEDGESET{v}}h^*_{e,v}\cdot y_e - \sum_{v \in V} m_v\cdot\eta_v} + (f-k)\cdot \sum_{v \in V} \sum_{e \in \INCEDGESET{v}}h^*_{e,v}\cdot y_e.
\end{align*}

\noindent
The former item is upper-bounded by $k\cdot\OP{OPT}(\Pi)$ by Lemma~\ref{lemma-dual-identical-feasible-region}.
Combing the above with Lemma~\ref{lemma-VC-RHC-opt-lower-bound}, we obtain
\begin{align*}
w(h^*) \enskip \le \enskip \BIGP{k+(f-k)\cdot\frac{k}{k-1}}\cdot\OP{OPT}(\Pi) \enskip = \enskip \BIGP{1+\frac{1}{k-1}}\cdot(f-1) \cdot\OP{OPT}(\Pi)
\end{align*}
as claimed.
\end{proof}

%

%


\section{Conclusion} 
\label{sec-conclusion}

We conclude with some future directions.
In this paper we presented bi-approximations for augmented multiplicity constraints.
It is also interesting to consider VC-HC with relaxed demand constraints, i.e., partial covers.
The reduction framework for partial VC-HC provided by Cheung et al.~\cite{doi:10.1137/1.9781611973402.124} and the tight approximation for VC-HC provided by Kao~\cite{Kao16VCHC} jointly provided an almost tight $f+\epsilon$-approximation when the vertices are unweighted.

\smallskip

When the vertices are weighted, it is known that $O\big(\frac{1}{\epsilon}\big) f$ bi-approximations can be obtained via simple LP rounding.
Comparing to the $O\big(\frac{1}{\epsilon}\big)$ bi-approximation result we can obtain for classical vertex cover, there is still a gap, and this would be an interesting direction to explore.
%


\bibliography{approx_capacitated_domination}

\newpage

\begin{appendix}

\normalsize


%

\begin{figure*}[htp]
\rule{\linewidth}{0.2mm}
\medskip
{{\sc Procedure} {\sc Primal-Dual}}
\begin{algorithmic}[1]
\STATE $w^\phi(v) \longleftarrow 1$, $d^\phi(v) \longleftarrow \sum_{e \in \INCEDGESET{v}}d_e$, for each $v \in V$.
\STATE $S \longleftarrow \emptyset$, $E^\phi \longleftarrow \{e: e\in E, d_e> 0 \}$, $V^\phi \longleftarrow \{v: v\in V, m_v\cdot c_v \cdot d^\phi(v)>0 \}$.
\WHILE{$E^\phi \neq \emptyset$}
	\STATE $r_v \longleftarrow w^\phi(v) / \min\{c_v, d^\phi(v)\}$, for each $v \in V^\phi$.
	\STATE $u \longleftarrow arg\min\{r_v: v \in V^\phi\}$. \enskip /* the next vertex to saturate */
	\STATE $w^\phi(v) \longleftarrow w^\phi(v) - r_u \cdot \min\{c_v, d^\phi(v)\}$, for each $v \in V^\phi$.
    \STATE $(S',h') \longleftarrow \OP{Self-Containment}(S \cup \{u\} , u)$.
    \IF{$S' = \phi$}
        \STATE $S \longleftarrow S \cup \{u\}$.
    \ELSE
        \STATE $S \longleftarrow S \setminus S'$.
        \STATE Remove $\INCEDGESET{S'}$ from $E^\phi$ and update $d^\phi(v)$ for each $v \in V^\phi$.
        \STATE For each $v \in V^\phi$ such that $d^\phi(v) = 0$, remove $v$ from $V^\phi$.
    \ENDIF    
	\STATE $V^\phi \longleftarrow V^\phi \setminus \{u\}$.
\ENDWHILE
\end{algorithmic}
\rule{\linewidth}{0.2mm}

\noindent\rule{\linewidth}{0.2mm}
{{\sc Procedure} $\OP{Self-Containment}(A,u)$}
\begin{algorithmic}[1]
\IF {$u \in A$}
	\STATE $\tilde{h} \longleftarrow \OP{Maxflow}(\MC{G}(A),u)$.  /* max-flow for $\MC{G}(A)$ such that $\tilde{h}(\tilde{u},s^-)$ is minimized. */
\ELSE
	\STATE $\tilde{h} \longleftarrow \OP{Maxflow}(\MC{G}(A))$.
\ENDIF
\STATE $S' \longleftarrow \{v \colon v \in A \text{ such that } \tilde{h}(s^+,\tilde{e}) = d_e$ for all $e \in \INCEDGESET{v} \cap E^\phi \}$.
\IF {$S' = A$ or $S' = \{\phi\}$}
	\STATE $\tilde{h'}_{e,v} \longleftarrow \tilde{h}(\tilde{e},\tilde{v})$, for all $v \in S'$ and $e \in \INCEDGESET{v}\cap E^\phi$.
    \STATE Return $(S', \tilde{h'})$.    
\ELSE
    \STATE Return $\OP{Self-Containment}(S',u)$.
\ENDIF
\end{algorithmic}
\rule{\linewidth}{0.2mm}

\caption{A pseudo-code for our Primal-Dual process.}
\label{algorithm_greedy_charging}
\end{figure*}

%

%


\section{A Primal-Dual Schema for VC-HC}

\medskip

\begin{ap_lemma}{\ref{lemma-primal-dual-feasibility}}
$E^\phi$ becomes empty in polynomial time.
Furthermore, the assignments computed by $\OP{Self-Containment}$ during the process form a feasible demand assignment.
\end{ap_lemma}

\begin{proof}
By procedure $\OP{Self-Containment}$, we know that all the edges in $\INCEDGESET{S'} \cap E^\phi$ will be removed from $E^\phi$ at the end of each iteration, where $S'$ is the set returned by $\OP{Self-Containment}$.

\smallskip

Hence, if $S \neq \emptyset$, then Lemma~\ref{lemma_local_fully_serve} guarantees that $\INCEDGESET{S} \cap E^\phi$ is not empty and we know that none of the vertices in $\bigcup_{e \in \INCEDGESET{S} \cap E^\phi}e$ was included in the set $S'$ returned by process $\OP{Self-Containment}$.
This further means that none of the vertices in $\bigcup_{e \in \INCEDGESET{S} \cap E^\phi}e \backslash S$ has saturated.
If none of them can saturate in later iterations, i.e., $\BIGP{\bigcup_{e \in \INCEDGESET{S} \cap E^\phi}e \backslash S} \cap V^\phi$ is empty, then we have found a proof that the input graph is infeasible since Lemma~\ref{lemma_local_fully_serve} guarantees that no self-containing subsets exist in $S$ after each iteration.

\smallskip

In other words, $\BIGP{\bigcup_{e \in \INCEDGESET{S} \cap E^\phi}e \backslash S} \cap V^\phi \neq \emptyset$ as long as the input graph is feasible.
Similarly, $E^\phi \neq \emptyset$ implies that $\BIGP{S \cup V^\phi} \cap \bigcup_{e \in E^\phi}e \neq \emptyset$.
%
%
%
%
Since the cardinality of $V^\phi$ strictly decreases in each iteration, we know that both $S$ and $E^\phi$ will eventually become empty if the input graph is feasible. 
%
%

\smallskip

The second half of this lemma follows from the fact that the demand assignment computed by $\OP{Self-Containment}$ is valid.
Since $E^\phi$ becomes empty eventually, the demand assignments computed in each iteration jointly form a feasible demand assignment for the input graph.
%
\end{proof}

\medskip

\medskip

\begin{ap_lemma}{\ref{claim-light-demands}}
For any $v \in V$ with $d^\phi(v) \le c_v$ when saturated, 
we can compute a function $\ell_v\colon \INCEDGESET{v} \rightarrow {\mathbb R}^{\ge 0}$ such that the following holds:
\begin{enumerate}[\quad (a)]
	\item
		$0\le h_{e,v} \le \ell_v(e) \le d_e$, for all $e \in \INCEDGESET{v}$.
	\item
		$\sum_{e \in \INCEDGESET{v}}\ell_v(e) \le c_v$.
	\item
		$\sum_{e \in \INCEDGESET{v}}\ell_v(e)\cdot y_e = w_v$.
\end{enumerate}
\end{ap_lemma}

\begin{proof}
Depending on the initial value of $d^\phi(v)$, we consider the following two cases.

\smallskip

If $d^\phi(v) \le c_v$ holds from the beginning, i.e., $\sum_{e\in \INCEDGESET{v}}d_e \le c_v$, then we set $\ell_v(e) = d_e$ for all $e \in \INCEDGESET{v}$.
As a result, condition~(a) and~(b) hold immediately.
For condition~(c), by our primal-dual scheme, we have $z_v = \eta_v = 0$ and $y_e = g_{e,v}$ for all $e \in \INCEDGESET{v}$ with $d_e > 0$. 
Therefore $w_v = \sum_{e \in \INCEDGESET{v}}d_e\cdot g_{e,v} = \sum_{e\in \INCEDGESET{v}}\ell_v(e)\cdot y_e$, and condition~(c) holds as well.

\smallskip

%
If $d^\phi(v) > c_v$ holds in the beginning, then consider the particular iteration of {\sc Dual-VCHC} for which $d^\phi(v)$ becomes less than or equal to $c_v$.
Let $H_v$ and $K_v$ denote the two sets of edges in $\INCEDGESET{v}$ that was removed from $E^\phi$ and that remained in $E^\phi$ in that iteration, respectively.
Note that $H_v \cap K_v = \emptyset$.

\smallskip

The function $\ell_v$ is defined by the following procedure:
For all $e \in K_v$, we set $\ell_v(e)$ to be $d_e$.
Let $c_v' = c_v - \sum_{e \in K_v}d_e$ be the remaining amount of demand to be collected.
We iterate over edges in $H_v$ and do the following for each $e \in H_v$:
\begin{itemize}
	\item
		If $c_v' \ge d_e$, then we set $\ell_v(e)$ to be $d_e$ and subtract $d_e$ from $c_v'$.

	\item
		Otherwise, we set $\ell_v(e)$ to be $c_v'$ and set $c_v'$ to be zero.
\end{itemize}


%
It follows that condition~(a) and (b) hold for the function $\ell_v$ defined above.
Below we show that condition~(c) holds as well.
By our primal-dual scheme, we know that $y_e = z_v + g_{e,v}$ holds for all $e \in H_v \cup K_v$.
(Note that the equality may not hold, however, for $e \in \INCEDGESET{v} \setminus (H_v \cup K_v)$.)
Furthermore, we know that $g_{e,v} = 0$ for all $e \in \INCEDGESET{v} \setminus K_v$.
Therefore, it follows that
\begin{align*}
w_v \enskip & = \enskip c_v\cdot z_v + \sum_{e \in \INCEDGESET{v}}d_e\cdot g_{e,v} \enskip = \enskip c_v\cdot z_v + \sum_{e \in K_v}d_e\cdot g_{e,v} \\[2pt]
& = \enskip \sum_{e \in K_v} d_e\cdot \BIGP{z_v+g_{e,v}} + \sum_{e \in H_v}\ell_v(e)\cdot z_v \\[3pt]
& = \enskip \sum_{e\in K_v} \ell_v(e) \cdot y_e + \sum_{e \in H_v}\ell_v(e)\cdot y_e \enskip = \enskip \sum_{e \in \INCEDGESET{v}}\ell_v(e)\cdot y_e,
\end{align*}
and condition~(c) holds as well.
This proves the lemma.
\end{proof}


%


\section{Augmented $\BIGP{k, (1+\frac{1}{k-1})(f-1)}$-Cover}

\medskip

\begin{ap_lemma}{\ref{lemma-VC-RHC-heavy-light-mixed-bound}}
We have $$\sum_{v \in V \setminus V_S}w_v\cdot x^{(h^*)}_v \le \BIGP{f-1}\cdot \sum_{v \in V_S}\sum_{e \in \INCEDGESET{v}}h^*_{e,v}\cdot y_e + f\cdot\sum_{v \in V \setminus V_S}\sum_{e \in \INCEDGESET{v}}h^*_{e,v}\cdot y_e.$$
\end{ap_lemma}

\begin{proof}
Consider any $v \in V \setminus V_S$ such that $x^{(h^*)}_v > 0$.
Depending on $c_v$, $\RECEIVED{h^*}{v}$, and $\RECEIVED{h}{v}$, we consider the following three exclusive cases separately:

\begin{enumerate}[(1)]
	\item
		If $\RECEIVED{h^*}{v} > c_v$, then we know that $\RECEIVED{h}{v} > c_v$. By Proposition~\ref{lemma_u_notin_s_heavy_unicost} we have 
		$$w_v\cdot x^{(h^*)}_v \enskip = \enskip c_v\cdot z_v \cdot \CEIL{\frac{\RECEIVED{h^*}{v}}{c_v}} \enskip \le \enskip 2\cdot \sum_{e\in \INCEDGESET{v}}h^*_{e,v}\cdot y_e.$$		
		In this case we charge the cost incurred by $v$ to the demand it serves, where each unit of demand, say, from edge $e$, gets a charge of $2\cdot y_e$.

	\item
		If $\RECEIVED{h}{v} > c_v \ge \RECEIVED{h^*}{v}$, then we know that $x^{(h^*)}_v = 1$.
		By Proposition~\ref{lemma_u_notin_s_heavy_unicost}, we have 
		$$w_v\cdot x^{(h^*)}_v \enskip = \enskip w_v \enskip = \enskip c_v\cdot z_v \enskip < \enskip \sum_{e \in \INCEDGESET{v}}h_{e,v}\cdot y_e,$$
		where the last inequality follows from the assumption that $\RECEIVED{h}{v} > c_v$.
		In this case we charge the cost of $v$ to the demand that was assigned to it in the original assignment $h$, where each unit demand gets a charge of $y_e$.

	\item
		If $c_v \ge \RECEIVED{h}{v}$, then we know that $h^*_{e,v} \le \ell_v(e)$ for all $e \in \INCEDGESET{v}$ by Lemma~\ref{claim-light-demands} and the way how $h^*$ is modified.
		Therefore, we have $x^{(h^*)}_v = 1$ and Lemma~\ref{claim-light-demands} states that
		$$w_v \cdot x^{(h^*)}_v = \sum_{e \in \INCEDGESET{v}}\ell_v(e)\cdot y_e.$$
		In this case, we charge the cost incurred by $v$ to the demand that is located in $\ell_v$, each of which gets a charge of $y_e$.
		%
		
		%
		%
		
		%
		%
		%
\end{enumerate}


\noindent
Consider any unit of demand from an edge $e \in E$ and the number of charges it gets in the above three cases.
Depending on the assignment $h^*$, we have the following three cases.
\begin{enumerate}[(a)]
	\item
		If the unit demand is assigned to a vertex in $V_S$, then it is charged at most $(f-1)$ times, i.e., at most once in case~(3) above by its remaining incident vertices.

	\item
		If the unit demand is assigned to a vertex $v \in V \setminus V_S$ with $\RECEIVED{h^*}{v} > c_v$, then it is charged twice in case~(1) above by $v$.

	\item
		If the unit demand is assigned to a vertex $v \in V \setminus V_S$ with $\RECEIVED{h^*}{v} \le c_v$, then it is charged at most $f$ times, i.e., at most once by all of its incident vertices in case~(2) and~(3) above.
\end{enumerate}

\noindent
Since $f \ge 2$, summing up the discussion and we obtain the statement as claimed.
\end{proof}

\end{appendix}

\end{document}